%% file: ms.tex
\documentclass[sigconf]{acmart}
\copyrightyear{2018} 
\acmYear{2018} 
\setcopyright{rightsretained}
\acmConference[PODC'18]{ACM Symposium on Principles of Distributed Computing}{July 23--27, 2018}{Egham, United Kingdom}
\acmDOI{10.1145/3212734.3212736}
\acmISBN{978-1-4503-5795-1/18/07} 

\input{defs}

\title{Atomic Cross-Chain Swaps}
\author{Maurice Herlihy}
\affiliation{%
  \department{Computer Science Department}
  \institution{Brown University}
  \city{Providence}
  \state{Rhode Island}
  \postcode{02912}
}

\email{maurice.herlihy@gmail.com}
\begin{abstract}
\input{abstract}  
\end{abstract}

\begin{document}
\maketitle

\input{podcbody}
\bibliographystyle{abbrv}
\bibliography{blockchain}

\end{document}

%% file: defs.tex
\usepackage{listings}
\usepackage{url}
\usepackage{amsmath, amsthm, amssymb}
\newcommand{\var}[1]{\text{\lstinline+#1+}}

      \def\cC{\ensuremath{{\mathcal C}}}   \def\cD{\ensuremath{{\mathcal D}}}
\def\cE{\ensuremath{{\mathcal E}}}      \def\cG{\ensuremath{{\mathcal G}}}

\newcommand{\set}[1]{\left\{ #1 \right\}}

\newcommand{\sig}{\ensuremath{\mathit{sig}}}

\newcommand{\diam}{\ensuremath{\mathit{diam}}}

\usepackage{epsfig}

\newcommand{\figlabel}[1]{\label{figure:#1}}
\newcommand{\nakedfigref}[1]{\ref{figure:#1}}

\newcommand{\Figref}[1]{Figure~\nakedfigref{#1}}

\newcommand{\lemmalabel}[1]{\label{lemma:#1}}
\newcommand{\nakedlemmaref}[1]{\ref{lemma:#1}}
\newcommand{\lemmaref}[1]{Lemma~\nakedlemmaref{#1}}

\newcommand{\thmlabel}[1]{\label{thm:#1}}
\newcommand{\nakedthmref}[1]{\ref{thm:#1}}
\newcommand{\thmref}[1]{Theorem~\nakedthmref{#1}}

\newcommand{\corlabel}[1]{\label{cor:#1}}

\newcommand{\linelabel}[1]{\label{line:#1}}
\newcommand{\nakedlineref}[1]{\ref{line:#1}}
\newcommand{\lineref}[1]{Line~\nakedlineref{#1}}
\newcommand{\linerefrange}[2]{Lines~\nakedlineref{#1}-\nakedlineref{#2}}

\newcommand{\kFreeride}{\textsc{FreeRide}}
\newcommand{\kDiscount}{\textsc{Discount}}
\newcommand{\kDeal}{\textsc{Deal}}
\newcommand{\kNodeal}{\textsc{NoDeal}}
\newcommand{\kUnderwater}{\textsc{Underwater}}
\newcommand{\bbP}{\ensuremath{\mathbb{P}}}

\graphicspath{{./pictures/}}

%% file: abstract.tex
An \emph{atomic cross-chain swap} is a distributed coordination task where multiple parties exchange assets across multiple blockchains, for example, trading bitcoin for ether.

An \emph{atomic} swap protocol guarantees (1) if all parties conform to the protocol, then all swaps take place, (2) if some coalition deviates from the protocol, then no conforming party ends up worse off, and (3) no coalition has an incentive to deviate from the protocol.

A cross-chain swap is modeled as a directed graph $\cD$, whose vertexes are parties and whose arcs are proposed asset transfers. For any pair $(\cD,L)$, where $\cD = (V,A)$ is a strongly-connected directed graph and $L \subset V$ a feedback vertex set for $\cD$, we give an atomic cross-chain swap protocol for $\cD$, using a form of hashed timelock contracts, where the vertexes in $L$ generate the hashlocked secrets. We show that no such protocol is possible if $\cD$ is not strongly connected, or if $\cD$ is strongly connected but $L$ is not a feedback vertex set. The protocol has time complexity $O(\diam(\cD))$ and space complexity (bits stored on all blockchains) $O(|A|^2)$.

%% file: podcbody.tex
\section{Motivation}
Carol wants to sell her Cadillac for bitcoins.
Alice is willing to buy Carol's Cadillac,
but she wants to pay in an ``alt-coin'' cryptocurrency.
Fortunately, Bob is willing to trade alt-coins for bitcoins.
Alice, Bob, and Carol need to arrange a three-way swap:
Alice will transfer her alt-coins to Bob,
Bob will transfer his bitcoins to Carol,
and Carol will transfer title of her Cadillac to
Alice\footnote{Naturally, they live in a state that records automobile
titles in a blockchain.}.
Of course, no one trusts anyone else.
How can we devise a protocol that ensures
that if all parties behave rationally,
in his or her own self-interest,
then all assets are exchanged,
but if some parties behave irrationally,
then no rational party will end up worse off?

In many blockchains,
assets are transferred under the control of so-called \emph{smart contracts}
(or just \emph{contracts}),
scripts published on the blockchain that establish and enforce
conditions necessary to transfer an asset from one party to another.
For example, let $H(\cdot)$ be a cryptographic hash function.
Alice might place her alt-coins in escrow by publishing on the alt-coin
blockchain a smart contract with \emph{hashlock} $h$ and \emph{timelock} $t$.
Hashlock $h$ means that if Bob sends the contract a value $s$,
called a \emph{secret},
such that $h = H(s)$,
then the contract irrevocably transfers ownership of those
alt-coins from Alice to Bob.
Timelock $t$ means that if Bob fails to produce that
secret before time $t$ elapses,
then the escrowed alt-coins are refunded to Alice.

\begin{figure*}[htb]
  \includegraphics[width=0.8 \hsize]{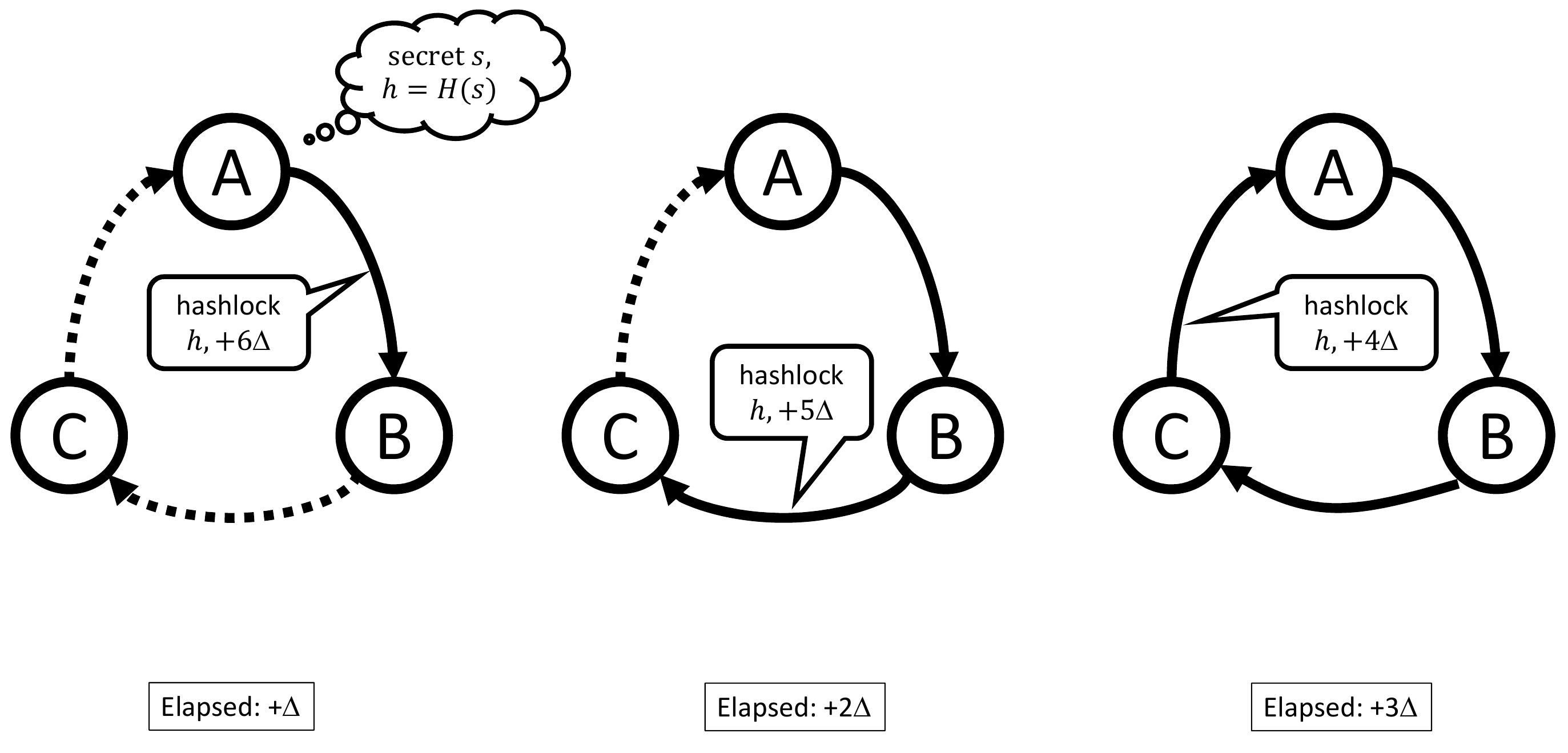}
  \caption{Atomic cross-chain swap: deploying contracts}
  \label{fig:phase1}
  \includegraphics[width=0.8 \hsize]{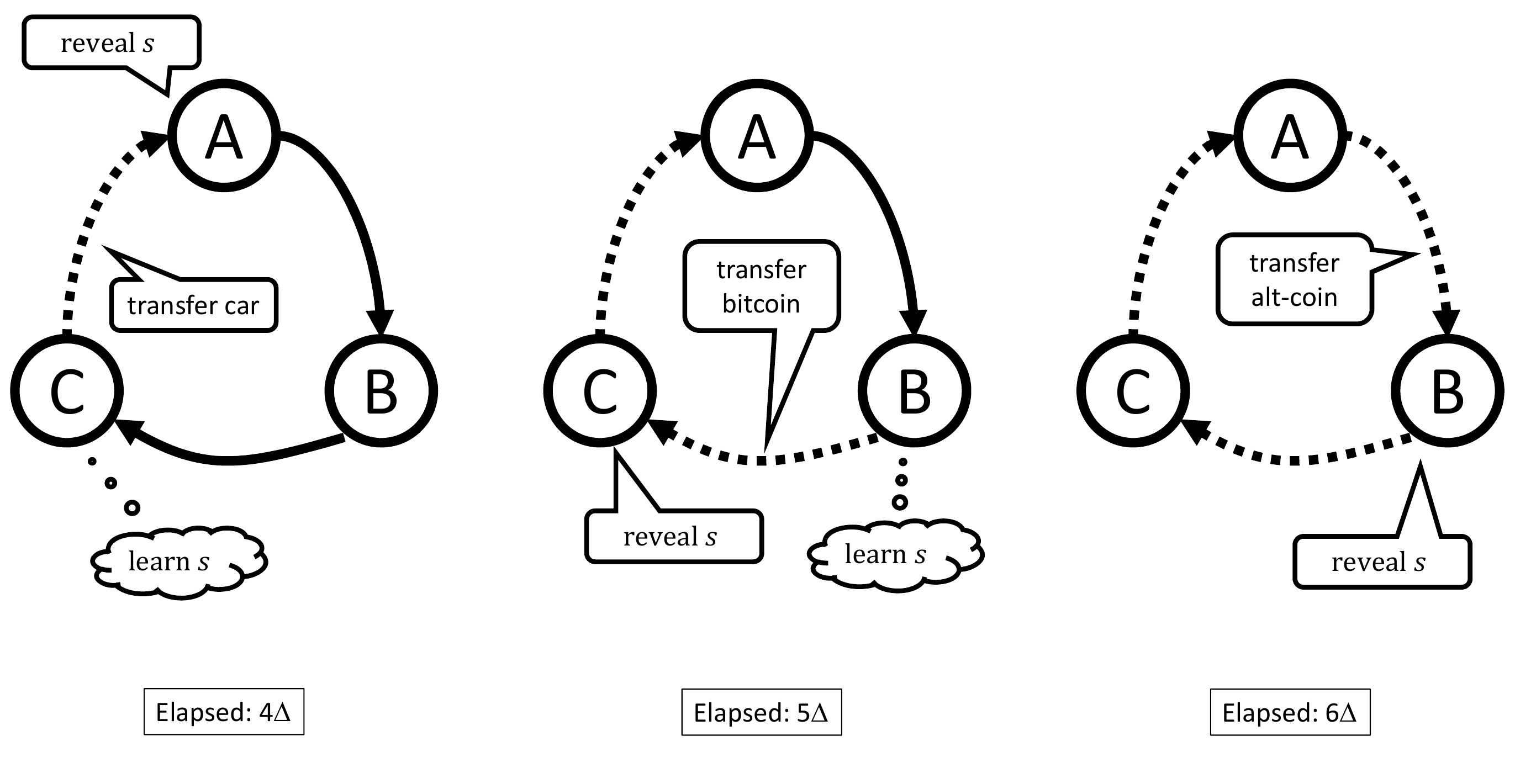}
  \caption{Atomic cross-chain swap: triggering arcs}
  \label{fig:phase2}
\end{figure*}
Here is a simple protocol for Alice, Bob, and Carol's three-way swap,
illustrated in Figures~\ref{fig:phase1} and~\ref{fig:phase2}.
Let $\Delta$ be enough time for one party to publish a smart contract
on any of the blockchains, or to change a contract's state,
and for the other party to detect the change.

\begin{itemize}
\item 
Alice creates a secret $s$, $h = H(s)$,
and publishes a contract on the alt-coin blockchain with hashlock $h$
and timelock $6\Delta$ in the future, to transfer her alt-coins to Bob.

\item
When Bob confirms that Alice's contract has been published on the
alt-coin blockchain,
he publishes a contract on the Bitcoin blockchain with the same
hashlock $h$ but with timelock $5\Delta$ in the future,
to transfer his bitcoins to Carol.

\item
When Carol confirms that Bob's contract has been published on the Bitcoin blockchain,
she publishes a contract on the automobile title blockchain with the
same hashlock $h$, but with timeout $4\Delta$ in the future,
to transfer her Cadillac's title to Alice.

\item
When Alice confirms that Carol's contract has been published on the
title blockchain, she sends $s$ to Carol's contract,
acquiring the title and revealing $s$ to Carol.

\item
Carol then sends $s$ to Bob's contract,
acquiring the bitcoins and revealing $s$ to Bob.

\item
Bob sends $s$ to Alice's contract,
acquiring the alt-coins and completing the swap.  
\end{itemize}
What could go wrong?
If any party halts while contracts are being deployed,
then all contracts eventually time out and trigger refunds.
If any party halts while contracts are being triggered,
then only that party ends up worse off.
For example,
if Carol halts without triggering her contract,
then Alice gets the Cadillac and Bob gets a refund,
so Carol's misbehavior harms only herself.

The order in which contracts are deployed matters.
If Carol were to post her contract with Alice
before Bob posts his contract with Carol,
then Alice could take ownership of the Cadillac without paying Carol.

Timelock values matter.
If Carol's contract with Bob were to expire at the same time as Bob's
contract with Alice,
then Carol could reveal $s$ to collect Bob's bitcoins at the very last moment,
leaving Bob no time to collect his alt-coins from Alice.

What if parties behave irrationally?
If Alice (irrationally) reveals $s$ before the first phase completes,
then Bob can take Alice's alt-coins,
and perhaps Carol can take Bob's bitcoins,
but Alice will not get her Cadillac,
so only she is worse off.

A \emph{atomic swap protocol} guarantees
(1) if all parties conform to the protocol, then all swaps take place,
(2) if some parties deviate from the protocol,
then no conforming party ends up
worse off\footnote{Other than the inconvenience of having assets temporarily
locked up.},
and (3) no coalition has an incentive to deviate from the protocol.
Alice, Bob, and Carol's swapping adventure suggests broader questions:
when are atomic cross-chain swaps possible,
how can we implement them,
and what do they cost?

While swapping digital assets is the immediate motivation for this study,
atomic cross-chain swap protocols have other possible applications.
\emph{Sharding}~\cite{sharding} splits one blockchain into many for
better load-balancing and scalability. 
Most of the time,
activities on different shards proceed independently.
When they cannot,
an atomic swap protocol can coordinate needed cross-chain updates.
In a decentralized distributed system,
\emph{upgrades} from one software version to another,
or from one data schema to another,
could benefit from atomic cross-chain swaps.
An atomic swap protocol can be thought of as a trust-free,
Byzantine-hardened form of \emph{distributed commitment}~\cite{twophasecommit}.
An atomic cross-chain swap is a special case of a
\emph{distributed atomic transaction}~\cite{WeikumV2001},
although not all atomic transactions can be expressed as cross-chain swaps.

Cross-chain swaps are well-known to the blockchain
community~\cite{bitcoinwiki,bip199,decred,tiersnolan,barterdex,Catalyst},
but to our knowledge,
this is the first systematic analysis of the theory underlying such protocols.
We make the following contributions.
A cross-chain swap is modeled as a directed graph (digraph) $\cD$,
whose vertexes are parties and whose arcs are proposed asset transfers.
For any pair $(\cD,L)$,
where $\cD = (V,A)$ is a \emph{strongly-connected} digraph
and $L \subset V$ a \emph{feedback vertex set} for $\cD$,
we give an atomic cross-chain swap protocol
using a form of hashed timelock contracts,
where the vertexes in $L$, called \emph{leaders},
generate the hashlocked secrets.
(Vertexes that are not leaders are \emph{followers}.)
We also show that no such protocol is possible if $\cD$ is not strongly connected,
or if $\cD$ is strongly connected
but the set of leaders $L$ is not a feedback vertex set.
The protocol has time complexity $O(\diam(\cD))$ and communication
complexity (bits published on blockchains) $O(|A| \cdot |L|)$.

\section{Model}
\subsection{Digraphs}
A \emph{directed graph} (or \emph{digraph}) $\cD$ is a pair $(V,A)$,
where $V$ is a finite set of \emph{vertexes},
and $A$ is a finite set of ordered pairs of distinct vertexes called
\emph{arcs}.
We use $V(\cD)$ for $\cD$'s set of vertexes,
and $A(\cD)$ for its set of arcs.
An arc $(u,v)$ has \emph{head} $u$ and \emph{tail} $v$.
An arc \emph{leaves} its head and \emph{enters} its tail.
An arc $(u,v)$ \emph{enters} a set of vertexes
$W \subseteq V$ if $u \not\in W$ and $v \in W$,
and similarly for leaving.

A digraph $\cC$ is a \emph{subdigraph} of $\cD$ if
$V(\cC) \subseteq V(\cD), A(\cC) \subseteq A(\cD)$
and every arc in $A(\cC)$ has both its head and tail in $V(\cC)$.

A \emph{path} $p$ in $\cD$ is a sequence of vertexes
$(u_0, \ldots, u_\ell)$ such that $u_0, \ldots, u_{\ell-1}$ are distinct.
Path $p$ has \emph{length} $\ell$, denoted by $|p|$.
If $v$ is a vertex,
and $(u_0, \ldots, u_\ell)$ a path that does not include the arc $(v,u_0)$,
then $v+p$ denotes the path $(v, u_0, \ldots, u_\ell)$.
For vertexes $u,v$,
$D(u,v)$ is the length of the longest path from $u$ to $v$ in $\cD$.

A path $(u_0, \ldots, u_\ell)$ is a \emph{cycle} if $u_0 = u_\ell$.
A digraph is \emph{acyclic} if it has no cycles.
Vertex $v$ is \emph{reachable} from vertex $u$ if there is a path from
$u$ to $v$.
$\cD$'s \emph{diameter} $\diam(\cD)$ is the length of
the longest path from any vertex to any other.
$\cD$ is \emph{connected} if its underlying graph is connected,
and \emph{strongly connected}
if, for every pair $u, v$ of distinct vertexes in $\cD$,
$u$ is reachable from $v$ and $v$ is reachable from $u$.
A \emph{feedback vertex set} is a subset of $V$
whose deletion leaves $\cD$ acyclic.

The \emph{transpose} $\cD^T$ is the digraph obtained
from $\cD$ by reversing all arcs.
If $\cD$ is strongly connected, so is $\cD^T$,
and any feedback vertex set for $\cD$ is also
a feedback vertex set for $\cD^T$.

\subsection{Blockchains and Smart Contracts}
For our purposes,
a \emph{blockchain} is a distributed service that allows clients to publish
transactions to a publicly-readable, tamper-proof distributed ledger.
Our analysis is independent of the particular blockchain algorithm.
We assume a timing model where there is a known duration $\Delta$ long
enough for one party to publish a contract to a blockchain,
and for a second party to confirm that the contract has been published.

The owner of an asset 
(such as a unit of cryptocurrency or an automobile title)
can create a smart contract to transfer ownership of that asset
to a \emph{counterparty} if specified conditions are met.
A contract is \emph{published} when its creator places it on a
blockchain ledger.
Once a contract is published,
it is irrevocable: neither the contract's creator
nor any other party can remove the contract nor tamper with its terms.

A \emph{rational} party acts in its own self-interest,
deviating from a protocol only if it is profitable to do so.
Rational parties can collude with one another to disadvantage other parties.
An \emph{irrational} party may deviate from a protocol even if it is
not profitable to do so.
Parties may behave irrationally out of spite,
because they were hacked,
or because they profit in ways not foreseen by the protocol designers.

Blockchain protocols typically require parties to have public and private keys.
We use $\sig(x,v)$ to denote the result of $v$ using its private key to sign $x$.

\section{Swap Digraphs and Games}
A cross-chain swap is given by a digraph $\cD = (V,A)$,
where each vertex in $V$ represents a party,
and each arc in $A$ represents a proposed asset transfer from the
arc's head to its tail via a shared blockchain.
(We assume without loss of generality that $\cD$ is connected,
because a disconnected digraph can be treated as multiple swaps.)
Henceforth, we use \emph{party} and \emph{vertex},
\emph{blockchain} and \emph{arc}, interchangeably,
depending on whether we emphasize roles or digraph structure.

In the terminology of game theory,
a swap $\cD$ is a \emph{cooperative game},
organized so that if all parties follow the protocol,
each transfer on each arc happens.
Each possible outcome is given by a subdigraph $\cE = (V,A')$ of $\cD$.
If a proposed transfer $(u,v) \in A$ is also in $(u,v) \in A'$,
then that transfer happened.
For short, we say arc $(u,v)$ was \emph{triggered}.

A protocol is a \emph{strategy} for playing a game:
a set of rules that determines which step a party
takes at any stage of a game.
To model real-world situations where multiple parties are secretly
controlled by a single adversary,  
the swap game is \emph{cooperative}:
parties can form \emph{coalitions} where coalition members
commit to a common strategy.

\begin{figure}[htb]
  \centering
  \includegraphics[width=\columnwidth]{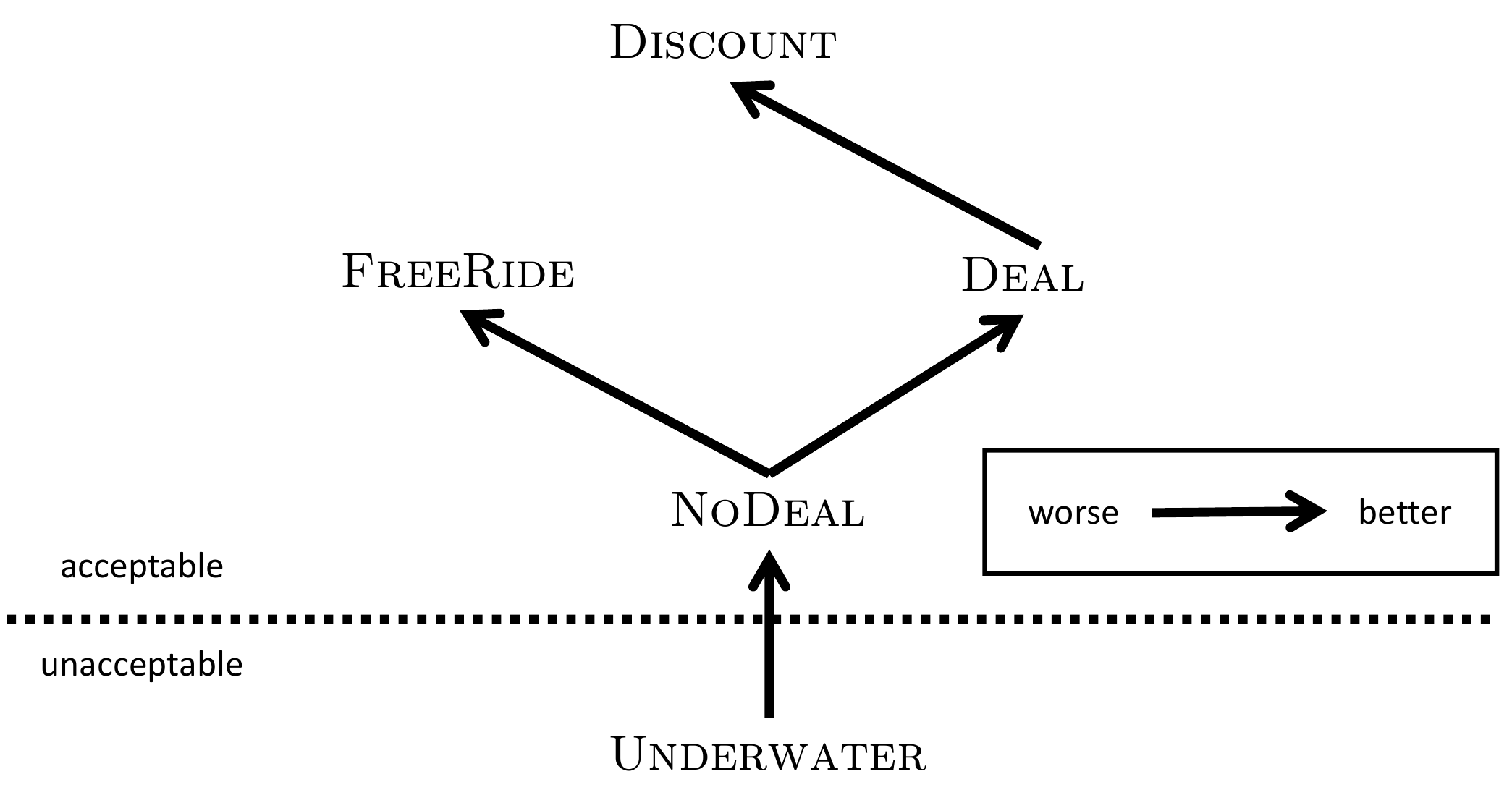}
  \caption{Partial order of protocol outcomes}
  \label{fig:outcomes}
\end{figure}
Here are the outcomes for a party $v$, organized into classes,
For brevity, each class has a shorthand name.
\begin{itemize}
\item
The party acquires assets without paying:
at least one arc entering $v$ is triggered,
but no arc leaving $v$ is triggered (\kFreeride).
  
\item
The party acquires assets while paying less than expected:
all arcs entering $v$ are triggered,
but at least one arc leaving $v$ is not triggered (\kDiscount).

\item
The party swaps assets as expected:
all arcs entering and leaving $v$ are triggered (\kDeal).

\item
No assets change hands:
no arc entering or leaving $v$ is triggered (\kNodeal).

\item
The party pays without acquiring all expected assets:
at least one arc entering $v$ is not triggered,
and at least one arc leaving $v$ is triggered
(\kUnderwater).  
\end{itemize}
Payoffs for a coalition $C \subset V$ are defined by replacing $v$
with $C$ in the definitions above.

The protocol design incorporates certain conservative assumptions
about parties' preferences. 
The protocol's preferred outcome is for all conforming parties to end with
outcome \kDeal{}.
In the presence of failures or deviation, however,
it is acceptable for conforming parties to end with outcome \kNodeal,
the \emph{status quo}, rendering them no worse off.
Furthermore,
each party is assumed to prefer \kDeal{} to \kNodeal{},
because otherwise it would not have agreed to the swap in the first place.
It follows that each party prefers any \kFreeride{} outcome to \kNodeal{},
because it acquires additional assets ``for free'',
without relinquishing any assets of its own.
Similarly,
each party prefers any \kDiscount{} outcome to \kDeal{},
since that party acquires the same assets in both outcomes,
but relinquishes fewer in \kDiscount{} outcomes.
For these reasons, \kDeal, \kNodeal, \kDiscount, and \kFreeride{} are
all considered acceptable outcomes for conforming parties if the
protocol execution is unable to complete because of failures or
adversarial behavior.

We consider the remaining \kUnderwater{} outcomes to be unacceptable
to conforming parties.
It is possible that in some idiosyncratic cases,
a party may actually prefer particular \kUnderwater{} outcomes to
\kNodeal{}.
For example,
a party with three entering arcs and one leaving arc
may be be willing to relinquish its asset in return for acquiring
only two out of three of the entering arcs' assets.
We leave the design of protocols that make such fine distinctions
to future work.

\begin{definition}
A swap protocol \bbP{} is \emph{uniform} if it satisfies:
\begin{itemize}
\item
If all parties follow \bbP{},
they all finish with payoff \kDeal.

\item
If any coalition cooperatively deviates from \bbP{},
no conforming party finishes with payoff \kUnderwater.
\end{itemize}
\end{definition}

A uniform protocol is not useful if rational parties will not follow it.
A swap protocol is a \emph{strong Nash equilibrium strategy} if no
coalition improves its payoff when its members cooperatively deviate
from that protocol.

\begin{definition}
A swap protocol \bbP{} is \emph{atomic} if it is both uniform and
a strong Nash equilibrium strategy.  
\end{definition}
This definition formalizes the notion that if all parties are rational,
all swaps happen,
but if some parties are irrational,
the rational parties will never end up worse off.
Recall that a \emph{conforming} party follows the protocol,
while a \emph{deviating} party does not.

\begin{lemma}
\lemmalabel{if}
If $\cD$ is strongly connected,
then any uniform swap protocol \bbP{} is atomic.
\end{lemma}

\begin{proof}
If a deviating coalition $C \subset V$ achieves a better payoff than \kDeal{},
then that payoff is either in \kFreeride{} or \kDiscount{}.
It follows that some arc that enters $C$ is triggered,
and some arc that leaves $C$ is untriggered.
Moreover, if any arc that enters $C$ is untriggered,
then all arcs that leave $C$ are untriggered.

A conforming party $v \not\in C$ cannot end up $\kUnderwater$,
so if an arc entering $v$ is untriggered,
then every arc leaving $v$ must be untriggered,
and if an arc leaving $v$ is triggered,
then every arc entering $v$ must be triggered.

Let $(c,v)$ be an untriggered arc leaving $C$.
Since $v$ is conforming,
every arc leaving $v$ is untriggered.
Because $\cD$ is strongly connected,
there is a path $(v,v_0),(v_1,v_2),\ldots,(v_k,c_0)$
where each $v_i \not\in C$, and $c_0 \in C$.
By a simple inductive argument,
each arc in this path is untriggered,
so the arc $(v_k,c_0)$ that enters $C$ is untriggered,
so \emph{every} arc leaving $C$ must be untriggered,
and some arc entering $C$ must be triggered.

Let $(v,c)$ be a triggered arc entering $C$.
Since $v$ is conforming,
every arc entering $v$ is triggered.
Because $\cD$ is strongly connected,
there is a path $(c_1,v_0), (v_0, v_1), \ldots, (v_k,c)$
where each $v_i \not\in C$, and $c_1 \in C$.
By a simple inductive argument,
each arc in this path is triggered,
so the arc $(c_1,v_0)$ leaving $C$ is triggered,
contradicting the fact that every arc leaving $C$ is untriggered.
\end{proof}
\begin{lemma}
\lemmalabel{only}
If $\cD$ is not strongly connected,
then no uniform swap protocol is atomic.
\end{lemma}

\begin{proof}
Because $\cD$ is not strongly connected,
it contains vertexes $x,y$ such that $y$ is reachable from $x$,
but not vice-versa.
Let $Y$ be the set of vertexes reachable from $y$,
and $X$ the rest: $X = V \setminus Y$.
$X$ is non-empty because it contains $x$.
Because $y$ is reachable from $x$,
there is at least one arc from $X$ to $Y$,
but no arcs from $Y$ to $X$.

Coalition $X$ can improve its payoff by triggering all
arcs between vertexes in $X$,
but no arcs from $X$ to $Y$,
yielding payoff \kFreeride{} for $X$,
since it triggers strictly fewer arcs leaving $X$,
without affecting any arcs entering $X$.
In fact,
the payoff for each individual vertex in $X$
is either the same or better than \kDeal.
\end{proof}
We have just proved:
\begin{theorem}
\thmlabel{nash}
A uniform swap protocol for $\cD$ is atomic
if and only if $\cD$ is strongly connected.
\end{theorem}
Informally, if $\cD$ is not strongly connected,
then rational parties will deviate from any uniform protocol.
In practice, such a swap would never be proposed,
because the parties in $X$ would never agree to a swap with
the free riders in $Y$. 
Henceforth, $\cD$ is assumed strongly connected.

We remark that \thmref{nash} relies on the implicit technical
assumption that all value transfers are explicitly represented on some
blockchain.
This theorem would be falsified, for example,
if Carol responds to learning Alice's secret by sending a
large drone to drop her Cadillac in the middle of Alice's driveway,
without ever recording that transfer in a shared blockchain.
We will assume that if swaps have off-chain consequences,
as they typically do,
that those consequences are explicitly recorded in the form of
blockchain updates.

\section{An Atomic Swap Protocol}
\begin{figure*}[htb]
\begin{lstlisting}[name=contract]
contract Swap {
  Asset asset;               /* asset to be transferred or refunded */`\linelabel{asset}\linelabel{state0}`
  Digraph digraph;           /* swap digraph */`\linelabel{digraph}`
  address[] leaders;         /* leaders */`\linelabel{leaders}`
  address party;             /* transfer asset from */`\linelabel{party}`
  address counterparty;      /* transfer asset to */`\linelabel{counterparty}`
  uint[] timelock;           /* vector of timelocks */`\linelabel{timelock}`
  uint[] hashlock;           /* vector of hashlocks */`\linelabel{hashlock}`
  bool[] unlocked;           /* which hashlocks unlocked? */`\linelabel{unlock}\linelabel{state1}`
  uint   start;              /* protocol starting time */
  /* constructor */
  function Swap (Asset _asset; /* asset to be transferred or refunded */`\linelabel{cons0}`
                 Digraph   _digraph;    /* swap digraph */
                 address[] _leaders;    /* leaders */
                 address   _party;      /* transfer asset from */
                 address   _counterparty; /* transfer asset to */
                 uint[]    _timelock;   /* vector of timelocks */
                 uint[]    _hashlock;   /* vector of hashlocks */
                 uint      _start       /* protocol starting time */
                 ) {
    asset = _asset;                               /* copy */`\linelabel{copy0}`
    party = _party; counterparty = _counterparty; /* copy */
    timelock = _timelock; hashlock = _hashlock;   /* copy */`\linelabel{copy1}`
    unlocked = [false, ..., false];               /* all unlocked */`\linelabel{init}`
  }`\linelabel{cons1}`
\end{lstlisting}
\caption{Swap contract (part one)}
\figlabel{contract0}
\end{figure*}

\begin{figure*}
\begin{lstlisting}[name=contract]
function unlock (int i, uint s, Path path, Sig sig) {`\linelabel{fun0}`
    require (msg.sender == counterparty); /* only from counterparty */`\linelabel{unlockreq}`
    if (now < start + (diam(digraph) + |path|) * `$\Delta$` /* hashkey still valid? */`\linelabel{timeOK}`
        && hashlock[i] == H(s)            /* secret correct? */`\linelabel{secretOK}`
        && isPath(path, digraph, leader[i], counterparty) /* path valid? */`\linelabel{pathOK}`
        && verifySigs(sig, s, path) { /* signatures valid? */`\linelabel{sigOK}`
      unlocked[i] = true;
    }
  }
  function refund () {`\linelabel{refund}`
    require (msg.sender == party); /* only from party */`\linelabel{claimreq}`
    if (any hashlock unlocked and timed out) {
      transfer asset to party;
      halt;
    }
  }
  function claim () {`\linelabel{claim}`
    require (msg.sender == counterparty); /* only from counterparty */
    if (every hashlock unlocked) {
      transfer asset to counterparty;
      halt;
    }
  }`\linelabel{fun1}`
}

\end{lstlisting}
\caption{Swap contract (part two)}
\figlabel{contract1}
\end{figure*}
\subsection{Hashlocks and Hashkeys}
In a simple two-party swap,
each party publishes a contract that assumes temporary control of that
party's asset.
This \emph{hashed timelock contract}~\cite{hashedtimelock}
stores a pair $(h,t)$,
and ensures that if the contract receives the matching secret $s$,
$h = H(s)$, before time $t$ has elapsed,
then the contract is \emph{triggered},
irrevocably transferring ownership of the asset to the counterparty.
If the contract does not receive the matching secret
before time $t$ has elapsed,
then the asset is \emph{refunded} to the original owner.
For multi-party cross-chain swaps,
we will need to extend these notions in several ways.

\begin{figure*}[htb]
  \centering
  \includegraphics[width=0.8 \hsize]{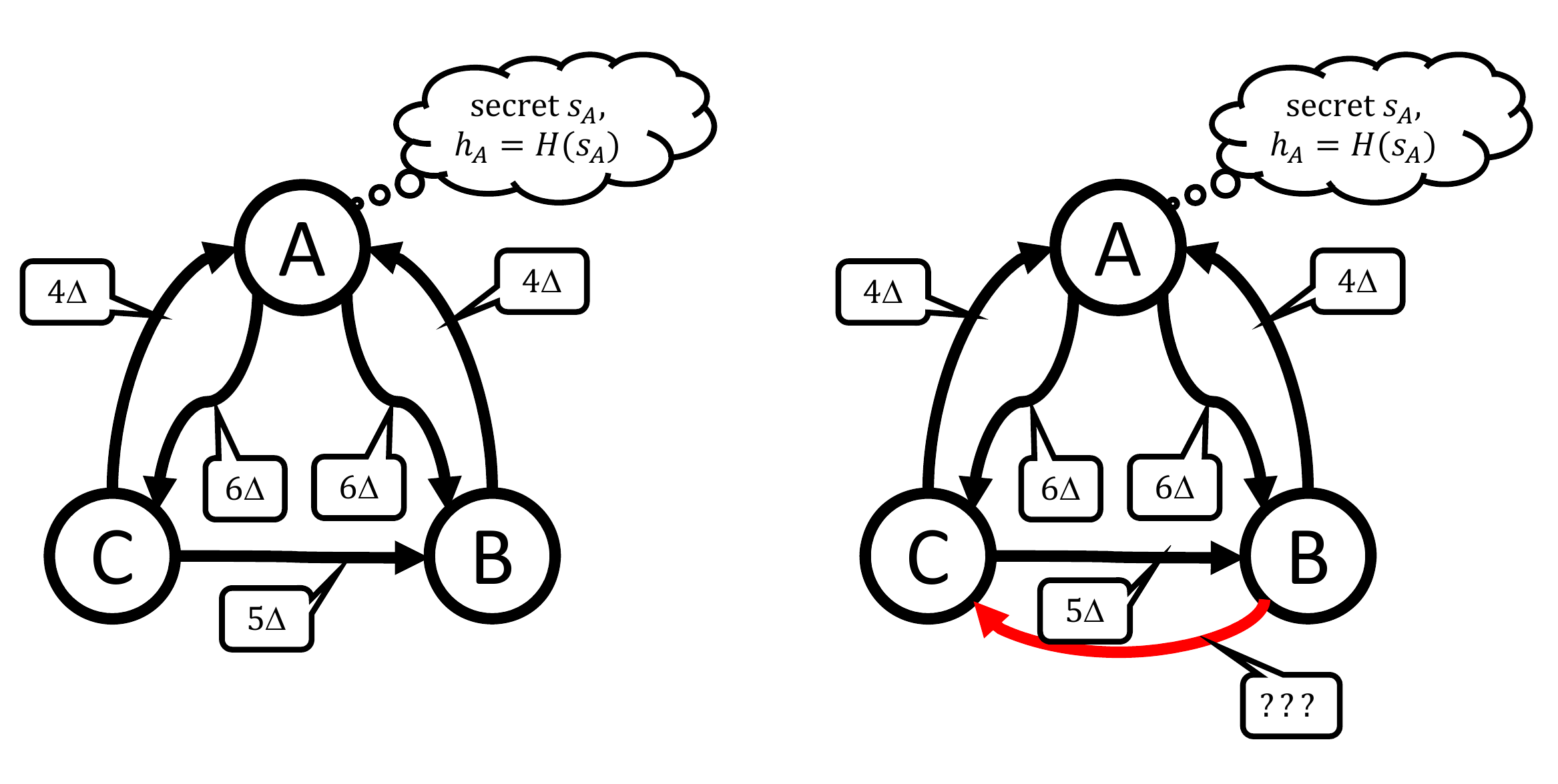}
  \caption{$A$ is the leader, $B$ and $C$ followers.
    Timeouts can be assigned when the follower subdigraph
    is acyclic (left) but not when it is cyclic (right)}
  \figlabel{timeouts}
\end{figure*}

In the three-way swap recounted earlier,
each arc had a single hashlock and a single timeout.
Timeouts were assigned so that the timeout on each arc entering a
follower $v$ was later by at least $\Delta$ than the timeout on each
arc leaving $v$.
This gap ensures that if any arc leaving $v$ is triggered,
then $v$ has time to trigger every entering arc.

If a swap digraph has only one leader, $\hat{v}$,
then the subdigraph of its followers is acyclic.
As in our three-way swap example,
the hashlock on arc $(u,v)$ can be given timeout
$(\diam(\cD) + D(v,\hat{v}) + 1) \cdot \Delta$,
where $D(v,\hat{v})$ is the length of the longest path
from $v$ to the unique leader $\hat{v}$.
(See left-hand side of \Figref{timeouts}.)

This formula does not work if a swap digraph has more than one leader,
because the subdigraph of any leader's followers has a cycle,
and it is not possible to assign timeouts across a cycle
in a way that guarantees a gap of at least $\Delta$ between entering
and leaving arcs.
(See right-hand side of \Figref{timeouts}.)

Instead, for general digraphs,
we must replace timed hashlocks with a more general mechanism,
one that assigns different timeouts to different paths.
Pick a set $L = \set{v_0,\ldots,v_\ell}$ of vertexes,
called \emph{leaders}, forming a feedback vertex set for $\cD$.
Each leader $v_i$ generates a secret $s_i$ and hashlock value $h_i = H(s_i)$,
yielding a \emph{hashlock vector} $(h_0, \ldots, h_\ell)$,
which is assigned to every arc.

\begin{figure*}[htb]
  \includegraphics[width=0.6 \hsize]{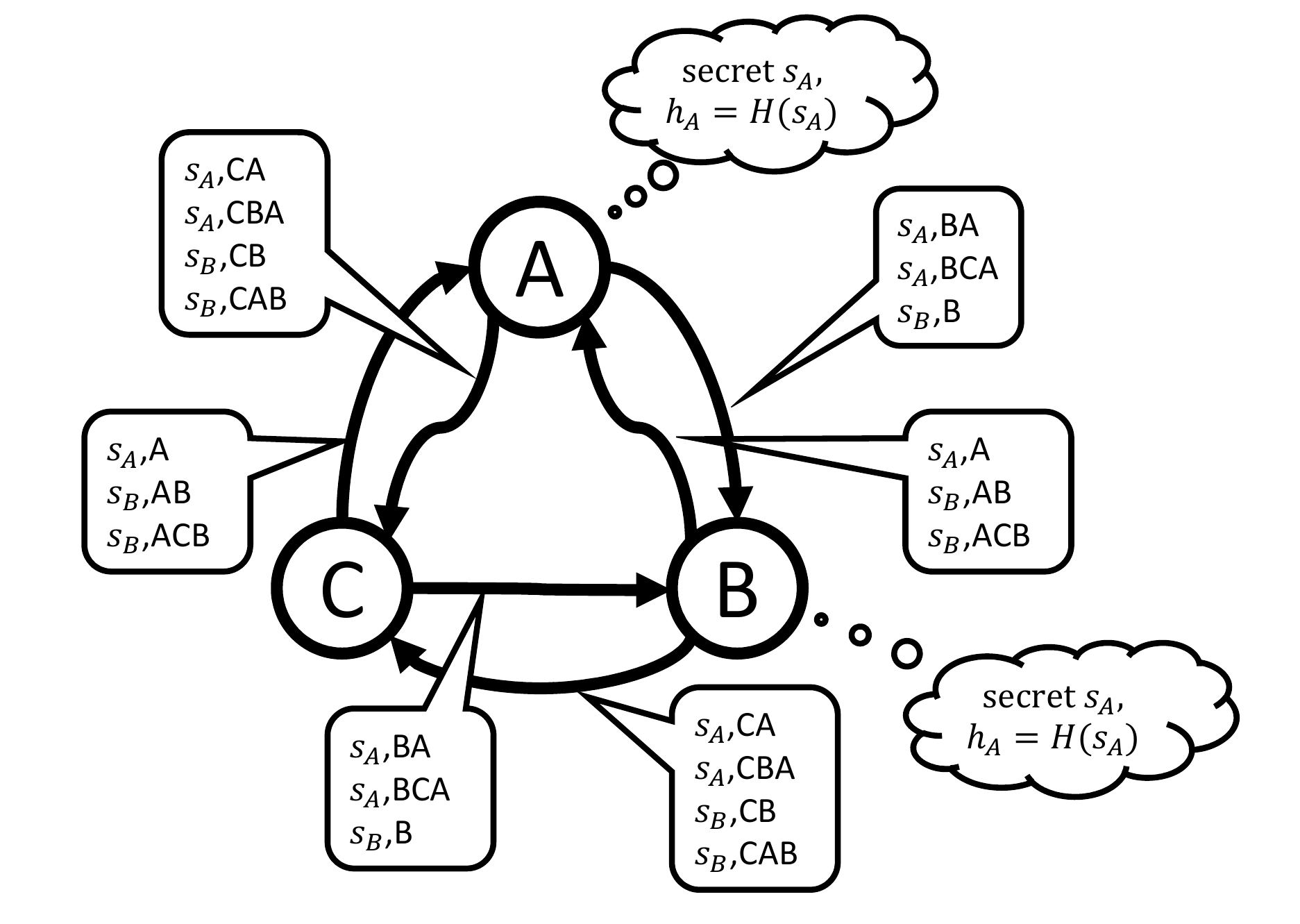}
  \caption{Hashkey paths for arcs of two-leader digraph}
\figlabel{hashkey}
\end{figure*}

\begin{figure*}[htb]
  \includegraphics[width=0.8 \hsize]{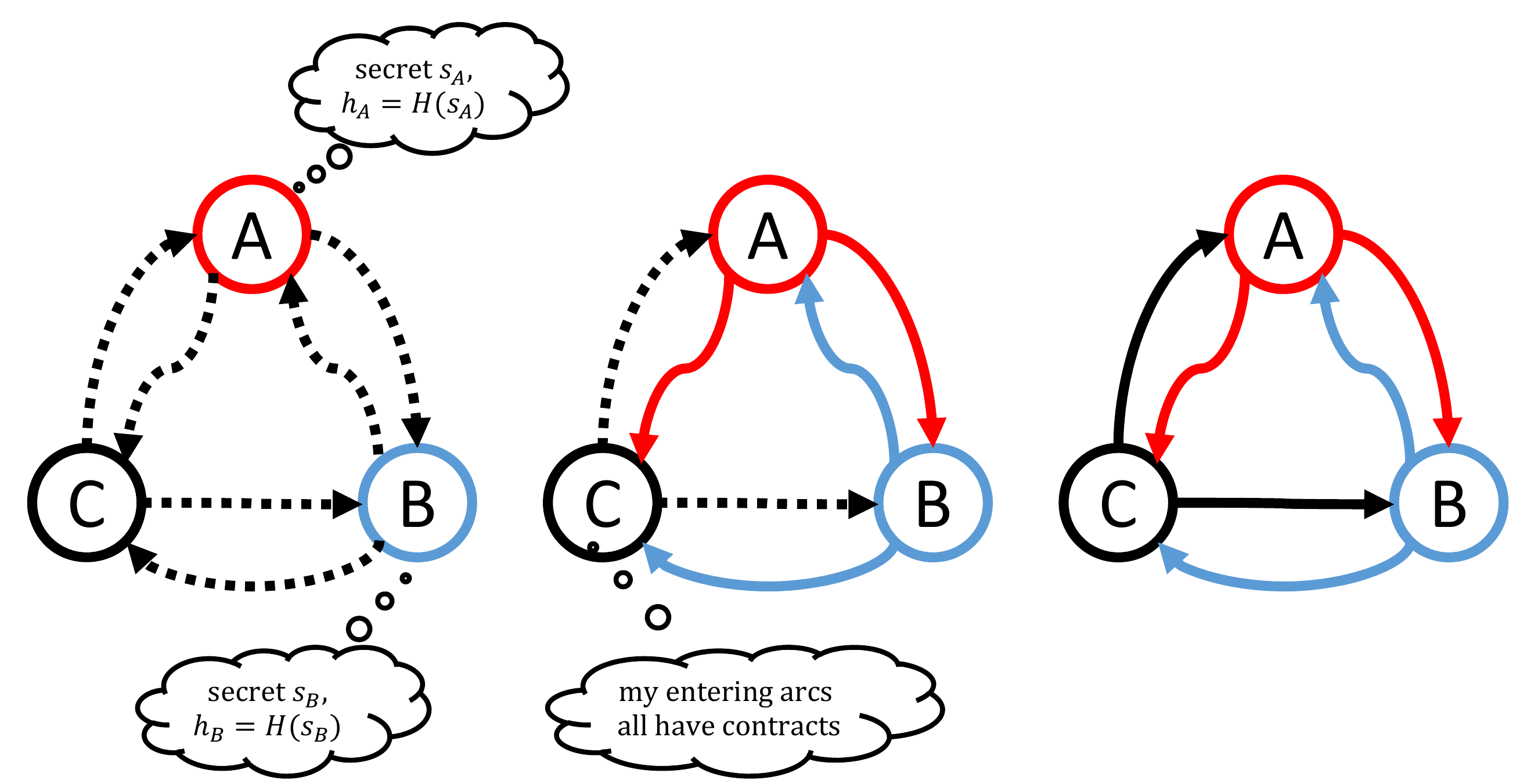}
  \caption{Concurrent contract propagation for two-leader digraph}
\figlabel{propagation}
\end{figure*}
A \emph{hashkey} for $h$ on arc $(u,v)$ is a triple $(s,p,\sigma)$,
where $s$ is the secret $h = H(s)$,
$p$ is a path $(u_0,\ldots,u_k)$ in $\cD$
where $u_0 = v$ and $u_k$ is the leader who generated $s$,
and
\begin{equation*}
\sigma = \sig(\cdots \sig(s, u_k), \ldots, u_0),
\end{equation*}
the result of having each party in the path sign $s$.
A hashkey $(s, p, \sigma)$ \emph{times out} at time
$(\diam(\cD) + |p|) \cdot \Delta$ after the start of the protocol.
That hashkey \emph{unlocks} $h$ on $(u,v)$ if it is presented before
it times out.
An arc is \emph{triggered} when all of its hashlocks are unlocked.
A hashlock has \emph{timed out} on an arc when all of its hashkeys
on that arc have timed out.
\Figref{hashkey} shows partial hashkeys for a two-leader swap digraph.

\subsection{Market Clearing}
For simplicity,
assume the swap digraph is constructed by a (possibly centralized)
market-clearing service,
which perhaps communicates with the parties through its own blockchain.
The clearing service is not a trusted party,
because the parties can check the consistency of the
clearing service's responses.

Each party creates a secret $s$ and matching hashlock $h = H(s)$.
It sends the clearing service its hashlock,
along with an offer characterizing the swaps it is willing to make.
The service combines these offers,
and publishes a swap digraph $\cD = (V,A)$,
a vector $L \subset V$ of \emph{leaders} forming a feedback vertex set,
a vector of those leaders' hashlocks $h_0, \ldots, h_\ell$,
and a \emph{starting time} $T$, at least $\Delta$ in the future.

If all parties conform to the protocol,
all contracts will be triggered before
$T + 2 \cdot \diam(\cD) \cdot \Delta$,
but if some parties deviate,
the conforming parties' assets will be refunded by then.

\subsection{Contracts}
Figures~\nakedfigref{contract0} and~\nakedfigref{contract1}
show pseudocode\footnote{
This pseudocode is based loosely on the popular Solidity
programming language for smart contracts~\cite{solidity}.
}
for a hashed timelock swap contract.
A smart contract resembles an object in an object-oriented programming language,
providing \emph{long-lived state} (\linerefrange{state0}{state1}),
a \emph{constructor} to  initialize that state
(\linerefrange{cons0}{cons1}),
and one or more \emph{functions} to manage that state
(\linerefrange{fun0}{fun1}).

The contract's long-lived state records
the asset to be transferred or refunded (\lineref{asset}),
the digraph $\cD$ (\lineref{digraph}),
the digraph's set of leaders (\lineref{leaders}),
the party transferring the asset (\lineref{party}),
the counterparty receiving the asset (\lineref{counterparty}),
a vector of timelocks (\lineref{timelock}),
a vector of hashlocks (\lineref{hashlock}),
and a Boolean \var{unlocked} vector marking which hashlocks have been
unlocked (\lineref{unlock}).

When the contract is initialized,
its constructor copies the fields provided into the contract's
long-lived state (\linerefrange{copy0}{copy1})
and sets each entry in \var{unlocked} to \emph{false} (\lineref{init}).

The \var{unlock()} function (\lineref{fun0}),
callable only by the counterparty (\lineref{unlockreq}),
takes an index $i$,
a secret $s_i$,
a path $p$,
and the signature \var{sig}.
The hashlock $h_i$ is unlocked if
\begin{itemize}
\item
  the current time is less than
  $T + (\diam(\cD) + |p|) \cdot \Delta$
  (\lineref{timeOK}), 
\item
  $h_i = H(s_i)$
  (\lineref{secretOK}),
\item
  $p$ is a path in $\cD$ from the counterparty to the leader who
  generated $s_i$
  (\lineref{pathOK}), and
\item
  the signature is the result of signing $s_i$ by the parties in $p$
  (\lineref{sigOK})
\end{itemize}

The \var{refund()} function (\lineref{refund}),
callable only by the party,
refunds the asset to the party if any unlocked hashlock has timed out.
The \var{claim()} function (\lineref{claim}),
callable only by the counterparty (\lineref{claimreq}),
transfers the asset to the counterparty if all hashlocks have been unlocked.

\subsection{Pebble Games}
We analyze the protocol using two variations on a simple pebble game.
We are given a strongly-connected digraph $\cD = (V,A)$,
and a vertex feedback set $L \subset V$ of \emph{leaders}.

In the \emph{lazy} pebble game,
start by placing pebbles on the arcs leaving each leader.
Place new pebbles on the arcs leaving vertex $v$ when there is a
pebble on \emph{every} arc entering $v$.

In the \emph{eager} pebble game,
start by placing a single pebble on one vertex $z$.
Place new pebbles on the arcs leaving $v$ when there is a pebble on
\emph{any} arc entering $v$.
Both games continue until no more pebbles can be placed.

\begin{lemma}
\lemmalabel{pebble:termlazy}
In the lazy game,
every arc in $\cD$ eventually has a pebble.
\end{lemma}

\begin{proof}
Suppose by way of contradiction,
the game stops in a state where an arc $(u,v)$ has no pebble.
There must be a pebble-free arc $(u',u)$ entering $u$,
because otherwise the game would have placed a pebble on $(u,v)$.
Continuing in this way,
build a longer and longer pebble-free path until it becomes a pebble-free cycle.
But leaders form a feedback vertex set,
so every cycle in $\cD$ includes a leader,
and the arcs leaving that leader have pebbles placed in the first step.
\end{proof}

\begin{lemma}
\lemmalabel{pebble:termeager}
In the eager game,
every arc in $\cD$ eventually has a pebble.
\end{lemma}

\begin{proof}
Suppose by way of contradiction,
the game stops in a state where an arc $(u,v)$ has no pebble.
Because $\cG$ is strongly connected,
there is a path from $z$ to $v$.
Since $z$ has a pebble and $v$ does not,
there is an arc $(w,w')$ on that path where $w$ has a pebble but $w'$ does not,
so $w'$ will get a pebble in the next step,
contradicting the hypothesis that the game has stopped.
\end{proof}

Suppose  there is a worst-case delay $\Delta$ between when the
last pebble is placed on any arc entering $v$,
and when the last pebble is placed on any arc leaving $v$.

\begin{lemma}
\lemmalabel{pebble:time}
In both pebble games,
every arc will have a pebble in time at most $\diam(\cD) \cdot
\Delta$ from when the game started.
\end{lemma}

\begin{proof}
For the lazy game,
it is enough to show that in each interval of time $\Delta$,
the longest pebble-free path shrinks by one.
At any time after the first step,
let $a_0,\ldots,a_k$ be a pebble-free path of maximal length.
That path cannot be a cycle,
because then it would include a leader,
who would have placed a pebble on $a_0$ in the first step.
It follows that every arc entering the head of $a_0$ must have a
pebble,
because otherwise we could construct a longer pebble-free path.
By hypothesis,
within time $\Delta$, $a_0$ will have a pebble,
and the path will have shrunk by one.

For the eager game,
it is enough to observe that in each interval of time $\Delta$,
for every vertex $v$,
the number of unpebbled vertexes in every path from $z$ to $w$ shrinks by one.
Because $\cD$ is strongly connected,
such a path always exists.
\end{proof}

\begin{corollary}
\corlabel{pebble}
Under the stated timing assumptions,
for both games,
every arc has a pebble within time $\diam(\cD) \cdot \Delta$.  
\end{corollary}

\subsection{The Protocol}
There are two phases.
In Phase One,
instances of the \var{Swap} contract
(Figures~\nakedfigref{contract0} and~\nakedfigref{contract1})
are propagated through $\cD$, starting at the leaders.
Each time a party observes that a contract has been published on an entering arc,
it verifies that contract is a correct swap contract,
and abandons the protocol otherwise.

Here is the Phase-One protocol for leaders:
\begin{enumerate}
\item
Publish a contract on every arc leaving the leader, then
\item
wait until contracts have been published on all arcs entering the leader.
\end{enumerate}
Here is the protocol for followers:
\begin{enumerate}
\item
wait until correct contracts have been published on all arcs entering
the vertex, then

\item
publish a contract on every arc leaving the vertex.
\end{enumerate}
\Figref{hashkey} shows how contracts are propagated in a swap digraph with two leaders.

In Phase Two, the parties disseminate secrets via hashkeys.
While contracts propagate in the direction of the arcs,
from party to counterparty,
hashkeys propagate in the opposite direction,
from counterparty to party.
Informally, each party is motivated to trigger the contracts
on entering arcs to acquire the assets controlled by those contracts.

We now trace how the secret $s_i$ generated by leader $v_i$ is
propagated.
At the start of the phase,
$v_i$ calls \var{unlock}$(s_i, v_i, \sig(s_i,v_i))$ at each entering arc's
contract
(here, the function's arguments are the hashkey, and $v_i$ is a degenerate path).
The first time any other party $v$ observes that hashlock $h_i$ on a leaving
arc's contract has been unlocked by a call to \var{unlock}$(s_i, p, \sigma)$,
it calls \var{unlock}$(s_i, v + p, \sig(\sigma,v))$ at each entering arc's contract.
The propagation of $s_i$ is complete when
$h_i$ has either timed out or has been unlocked on all arcs.

\begin{lemma}
\thmlabel{phase1}
If all parties conform to the protocol,
then every arc has a contract within time
$\diam(\cD) \cdot \Delta$ of when the protocol started.
\end{lemma}

\begin{proof}
  Phase One is an instance of the lazy pebble game on $\cD$,
  so the claim follows from
  Lemmas~\nakedlemmaref{pebble:termlazy} and~\nakedlemmaref{pebble:time}.
\end{proof}

\begin{lemma}
\thmlabel{phase2}
If all parties conform to the protocol,
then every arc's contract is triggered within time
$2 \cdot \diam(\cD) \cdot \Delta$ of when the protocol started.
\end{lemma}

\begin{proof}
Each secret's dissemination is an instance of the eager pebble games
on $\cD^T$, the transpose digraph.
The secrets are disseminated in parallel.
\end{proof}

\begin{theorem}
\thmlabel{time}
If all parties conform to the protocol,
then every contract is triggered within time
$2 \cdot \diam(\cD) \cdot \Delta$ of when the protocol started.
\end{theorem}

The deadline $2 \cdot \diam(\cD) \cdot \Delta$ bounds the time assets can be held
in escrow when things go wrong.
In practice,
one would expect actual running times to be shorter.

There is a simple optimization that ensures that Phase Two
completes in constant time when all parties conform to the protocol.
We use a shared blockchain,
perhaps that of the market-clearing service,
as a broadcast medium.
Each leader $v_i$ publishes its secret $s_i$ on the shared blockchain,
and each follower monitors that blockchain,
triggering its entering arcs when it learns the secret.
(Logically, we create an arc from each follower directly to that leader.)
Unfortunately, while this broadcasting blockchain can
``short-circuit'' the Phase Two protocol,
it cannot replace it,
because a deviating leader might refrain from publishing the secret
on that blockchain, but publish it on others.
(Miller \emph{et al.}~\cite{MillerBKM17} propose a similar optimization
for the Lightning network.)

\begin{lemma}
\lemmalabel{hashkey}
If hashlock $h$ times out on any arc entering a conforming $v$,
then $h$ must have timed out on every arc leaving $v$.
\end{lemma}

\begin{proof}
  Suppose $h$ was triggered on
  $(v,w)$ by hashkey $(s,p,\sigma)$.
  If $v$ does not appear in $p$,
  then $v + p$ is a path from $v$ to the leader,
  and $v$ can immediately trigger $h$ on $(u,v)$ using the
  hashkey $(s, v+p, \sig(\sigma,v))$, which has not timed out.
  If $v$ appears in $p$,
  then $v$ has already received (and signed) a hashkey that triggers $h$ on $(u,v)$.
\end{proof}

\begin{theorem}
  \thmlabel{main}
No conforming party ends up \kUnderwater{}.
\end{theorem}

\begin{proof}
Assume by way of contradiction that some conforming party $v$ ends up \kUnderwater{}:
a leaving arc $(v,w)$ has a triggered contract,
but an entering arc $(u,v)$ does not and will not.
  
First, arc $(u,v)$ must have a contract.
Suppose $v$ is a leader.
Since $(v,w)$ has been triggered,
$v$ has revealed its secret through a hashkey.
But a leader issues hashkeys in Phase Two
only after a contract has been published on every entering arc during Phase One.
Suppose instead $v$ is a follower.
Since $(v,w)$ has been triggered,
one of the arcs leaving $v$ has a contract.
But in Phase One, a follower publishes a contract on a leaving arc only after
contracts have been published on all of its incoming arcs.

Since $(u,v)$ has a contract,
one of that arc's hashlocks must have timed out.
By \lemmaref{hashkey},
the arc $(v,w)$ must also have timed out,
a contradiction.
\end{proof}

\begin{theorem}
\thmlabel{comm}
For $\cD = (V,A)$ with leaders $L \subset V$,
the space complexity,
measured as the number of bits stored on all blockchains,
is $O(|A|^2)$.
\end{theorem}

\begin{proof}
  There are $|A|$ contracts, one on each arc,
  each with a copy of the digraph $\cD$,
  which requires $O(|A|)$ storage.
\end{proof}

Finally, any atomic cross-chain swap protocol using hashed timelocks
must assign secrets to a feedback vertex set.
\begin{lemma}
\lemmalabel{publish}
In any uniform hashed timelock swap protocol,
no follower $v$ can publish a contract on an arc leaving $v$
before contracts have been published on all arcs entering $v$.
\end{lemma}

\begin{proof}
If follower $v$ has has a contract on arc $(v,w)$ but no contract on arc $(u,v)$,
then the parties other than $v$ could collude to trigger the contract on $(v,w)$,
while refusing to publish a contract on $(u,v)$,
leaving $v$ \kUnderwater.
\end{proof}

\begin{theorem}
In any uniform swap protocol based on hashed timelocks,
the set $L$ of leaders is a feedback vertex set in $\cD$.
\end{theorem}

\begin{proof}
Suppose, instead,
there is a uniform swap protocol where the leaders
do not form a vertex feedback set.

At any step in the protocol,
the \emph{waits-for} digraph $W$ is the subdigraph of $\cD^T$ where
$(v,u)$ is an arc of $W$ if $(u,v)$ has no published contract.
Informally,
\lemmaref{publish} implies that $v$ must be waiting for $u$ to publish
a contract on $(u,v)$ before $u$ can publish any contracts on its own
outgoing arcs.
In the initial state, if $\cD \setminus L$ contains a cycle, so does $W$.
At each protocol step,
a follower $v$ can publish a contract on a leaving arc only if $v$
has indegree zero in the current waits-for digraph.
But no vertex on a cycle in the waits-for digraph will ever have indegree zero,
a contradiction.
\end{proof}

\subsection{Single-Leader Digraphs}
As noted, in the common special case where a swap digraph needs only
one leader,
we can replace hashkeys with simple timeouts,
reducing message sizes and eliminating the need for digital
signatures.
In the following, let $\cD$ be a swap digraph with a single leader
$\hat{v}$ with hashlock $h$.
\begin{lemma}
  \lemmalabel{simple}
  If each arc $(u,v)$ has timeout
  $(\diam(\cD) + D(v,\hat{v}) + 1) \cdot \Delta$,
  then for every conforming $v \neq \hat{v}$,
  the timeout on each arc $(u,v)$ is later by at least $\Delta$ than
  the timeout on each arc $(v,w)$.
\end{lemma}

\begin{proof}
  Let $p$ be the longest path from $w$ to the leader $\hat{v}$.
  Because the subdigraph of followers is acyclic,
  $v + p$ is a path of length $D(w,\hat{v})+1$ from $v$ to
  $\hat{v}$, so $D(v,\hat{v}) \geq D(w, \hat{v})+1$.
\end{proof}

\begin{lemma}
  \lemmalabel{timeout}
  For a single-leader digraph using timeouts,
  if hashlock $h$ times out on any arc entering a conforming $v$,
  then $h$ must have timed out on every arc leaving $v$.
\end{lemma}

\begin{proof}
  By \lemmaref{timeout},
  once $h$ is triggered on $(v,w)$,
  $v$ has time at least $\Delta$ to trigger $(u,v)$.
\end{proof}

From this point on,
the bounds on running time and proofs of safety for the
single-leader-using-timeouts protocol are essentially the same as for
the general protocol.

\section{Remarks}
We have seen that single-leader swap digraphs do not require hashkeys
and digital signatures, only timeouts.
Is there a way to reduce the use of digital signatures in the general case?

Finding a minimal feedback vertex set for $\cD$ is NP-complete~\cite{Karp72},
although there exists an efficient 2-approximation~\cite{BeckerG1996}.

The protocol is easily extended to a model where there
may be more than one arc from one vertex to another,
so-called \emph{directed multi-graphs}~\cite{bang2001digraphs},
reflecting the situation where Alice wants to transfer
assets on distinct blockchains to Bob.

The swap protocol is still vulnerable to a weak denial-of-service
attack where an adversarial party repeatedly proposes an attractive
swap, and then fails to complete the protocol,
triggering refunds, but temporarily rendering assets inaccessible.
We leave for future work the question whether one could require
parties to post bonds, and following a failed swap.
examine the blockchains to determine who was at fault
(by failing to execute an enabled transition).

An interesting open problem is the extent to which this swap
protocol can be modified to provide better privacy,
analogous to the way the Bolt network~\cite{bolt} improves on Lightning.

As noted,
some parties may be willing to accept certain \kUnderwater{} outcomes
rejected by the swap protocol presented here.
Future work might investigate protocols where parties are endowed with
customized objective functions to provide finer-grained control which
outcomes are acceptable.

The swap protocol can be made recurrent
by having the leaders distribute the next round's hashlocks in Phase
Two of the previous round.
If swaps are recurrent,
then it would be useful to conduct swaps \emph{off-chain} as much as possible,
similar to the way that Lightning~\cite{lightning} and
Raiden~\cite{raiden} networks support off-chain transactions for
bitcoin and ERC20 tokens.

One limitation of the swap protocol presented here is the
assumption that the swap digraph, its leaders, and their hashlocks are
common knowledge among the participants.
Future work might address constructing and propagating this
information dynamically.

\section{Related Work}
The use of hashed timelock contracts for two-party cross-chain swaps
is believed to have emerged from an on-line discussion forum in
2016~\cite{bitcoinwiki,tiersnolan}.
There is open-source code~\cite{bip199,decred,barterdex} for two-party
cross-chain swap protocols between selected currencies,
and proposals for applications using swaps~\cite{Catalyst}.

Off-chain payment networks~\cite{DeckerW2015,raiden,bolt,lightning} circumvent the
scalability limits of existing blockchains by conducting multiple
transactions off the blockchain,
eventually resolving final balances through
a single on-chain transaction.
The Revive network~\cite{KhalilG2017} rebalances off-chain networks in
a way that ensures that compliant parties do not end up worse off.
These algorithms also use hashed timelock contracts,
but they address a different set of problems.

Multi-party swaps arise when matching kidney donors and recipients.
A transplant recipient with an incompatible donor can swap donors to
ensure that each recipient obtains a compatible organ.
A number of algorithms~\cite{AbrahamBS2007,DickersonMPST16,Jia2017}
have been proposed for matching donors and recipients.
Shapley and Scarf~\cite{ShapleyS1974} consider the circumstances under
which certain kinds of swap markets have strong equilibriums.
Kaplan~\cite{Kaplan2011} describes a polynomial-time algorithm that
given a set of proposed swaps,
constructs a swap digraph if one exists.
These papers and many others focus on ``the clearing problem'',
roughly analogous to constructing a swap digraph,
but not on how to execute those swaps on blockchains.

The \emph{fair exchange} problem~\cite{FranklinT1998,Micali2003} is a
precursor to the atomic cross-chain swap problem.
Alice has a digital asset Bob wants, and vice-versa,
and at the end of the protocol,
either Alice and Bob have exchanged assets,
or they both keep their assets.
In the absence of blockchains,
trusted, or semi-trusted third parties are required,
but roles of those trusted parties can be minimized in clever ways.

A \emph{atomic cross-chain transaction} is a distributed task where a
\emph{sequence} of exchanges occurs at each blockchain.
An atomic cross-chain swap is an atomic cross-chain transaction,
but not vice-versa,
because not all transactions can be expressed as swaps.
In our original example,
Alice could not borrow bitcoins from Bob to pay Carol,
because then Alice would have to execute two steps in sequence
(borrow, then spend) instead of executing a single swap.
A better understanding of atomic cross-chain transactions is the
subject of future work.

%% file: ms.bbl
\begin{thebibliography}{10}

\bibitem{AbrahamBS2007}
D.~J. Abraham, A.~Blum, and T.~Sandholm.
\newblock Clearing algorithms for barter exchange markets: Enabling nationwide
  kidney exchanges.
\newblock In {\em Proceedings of the 8th ACM Conference on Electronic
  Commerce}, EC '07, pages 295--304, New York, NY, USA, 2007. ACM.

\bibitem{bang2001digraphs}
J.~Bang-Jensen and G.~Gutin.
\newblock {\em Digraphs: Theory, Algorithms, and Applications}.
\newblock Monographs in Mathematics. Springer, 2001.

\bibitem{BeckerG1996}
A.~Becker and D.~Geiger.
\newblock Optimization of pearl's method of conditioning and greedy-like
  approximation algorithms for the vertex feedback set problem.
\newblock {\em Artificial Intelligence}, 83(1):167 -- 188, 1996.

\bibitem{bitcoinwiki}
bitcoinwiki.
\newblock Atomic cross-chain trading.
\newblock \url{https://en.bitcoin.it/wiki/Atomic_cross-chain_trading}.
\newblock As of 9 January 2018.

\bibitem{hashedtimelock}
bitcoinwiki.
\newblock Hashed timelock contracts.
\newblock \url{https://en.bitcoin.it/wiki/Hashed_Timelock_Contracts}.
\newblock As of 8 January 2018.

\bibitem{bip199}
S.~Bowe and D.~Hopwood.
\newblock Hashed time-locked contract transactions.
\newblock \url{https://github.com/bitcoin/bips/blob/master/bip-0199.mediawiki}.
\newblock As of 9 January 2018.

\bibitem{sharding}
V.~Buterin.
\newblock On sharding blockchains.
\newblock \url{https://github.com/ethereum/wiki/wiki/Sharding-FAQ}.
\newblock As of 8 January 2018.

\bibitem{DeckerW2015}
C.~Decker and R.~Wattenhofer.
\newblock A fast and scalable payment network with bitcoin duplex micropayment
  channels.
\newblock In A.~Pelc and A.~A. Schwarzmann, editors, {\em Stabilization,
  Safety, and Security of Distributed Systems}, pages 3--18, Cham, 2015.
  Springer International Publishing.

\bibitem{decred}
DeCred.
\newblock Decred cross-chain atomic swapping.
\newblock \url{https://github.com/decred/atomicswap}.
\newblock As of 8 January 2018.

\bibitem{DickersonMPST16}
J.~P. Dickerson, D.~F. Manlove, B.~Plaut, T.~Sandholm, and J.~Trimble.
\newblock Position-indexed formulations for kidney exchange.
\newblock {\em CoRR}, abs/1606.01623, 2016.

\bibitem{FranklinT1998}
M.~K. Franklin and G.~Tsudik.
\newblock Secure group barter: Multi-party fair exchange with semi-trusted
  neutral parties.
\newblock In {\em Financial Cryptography}, 1998.

\bibitem{bolt}
M.~Green and I.~Miers.
\newblock Bolt: Anonymous payment channels for decentralized currencies.
\newblock Cryptology ePrint Archive, Report 2016/701, 2016.
\newblock \url{https://eprint.iacr.org/2016/701}.

\bibitem{Jia2017}
Z.~Jia, P.~Tang, R.~Wang, and H.~Zhang.
\newblock Efficient near-optimal algorithms for barter exchange.
\newblock In {\em Proceedings of the 16th Conference on Autonomous Agents and
  MultiAgent Systems}, AAMAS '17, pages 362--370, Richland, SC, 2017.
  International Foundation for Autonomous Agents and Multiagent Systems.

\bibitem{Kaplan2011}
R.~M. Kaplan.
\newblock An improved algorithm for multi-way trading for exchange and barter.
\newblock {\em Electronic Commerce Research and Applications}, 10(1):67 -- 74,
  2011.
\newblock Special Section: Service Innovation in E-Commerce.

\bibitem{Karp72}
R.~M. Karp.
\newblock Reducibility among combinatorial problems.
\newblock In {\em Proceedings of a symposium on the Complexity of Computer
  Computations, held March 20-22, 1972, at the {IBM} Thomas J. Watson Research
  Center, Yorktown Heights, New York.}, pages 85--103, 1972.

\bibitem{KhalilG2017}
R.~Khalil and A.~Gervais.
\newblock Revive: Rebalancing off-blockchain payment networks.
\newblock In {\em Proceedings of the 2017 ACM SIGSAC Conference on Computer and
  Communications Security}, CCS '17, pages 439--453, New York, NY, USA, 2017.
  ACM.

\bibitem{Micali2003}
S.~Micali.
\newblock Simple and fast optimistic protocols for fair electronic exchange.
\newblock In {\em Proceedings of the Twenty-second Annual Symposium on
  Principles of Distributed Computing}, PODC '03, pages 12--19, New York, NY,
  USA, 2003. ACM.

\bibitem{MillerBKM17}
A.~Miller, I.~Bentov, R.~Kumaresan, and P.~McCorry.
\newblock Sprites: Payment channels that go faster than lightning.
\newblock {\em CoRR}, abs/1702.05812, 2017.

\bibitem{raiden}
R.~Network.
\newblock What is the raiden network?
\newblock \url{https://raiden.network/101.html}.
\newblock As of 26 January 2018.

\bibitem{tiersnolan}
T.~Nolan.
\newblock Atomic swaps using cut and choose.
\newblock \url{https://bitcointalk.org/index.php?topic=1364951}.
\newblock As of 9 January 2018.

\bibitem{barterdex}
T.~K. Organization.
\newblock The barterdex whitepaper: A decentralized, open-source cryptocurrency
  exchange, powered by atomic-swap technology.
\newblock
  \url{https://supernet.org/en/technology/whitepapers/BarterDEX-Whitepaper-v0.4.pdf}.
\newblock As of 9 January 2018.

\bibitem{lightning}
J.~Poon and T.~Dryja.
\newblock The bitcoin lightning network: Scalable off-chain instant payments.
\newblock \url{https://lightning.network/lightning-network-paper.pdf}, Jan.
  2016.
\newblock As of 29 December 2017.

\bibitem{ShapleyS1974}
L.~Shapley and H.~Scarf.
\newblock On cores and indivisibility.
\newblock {\em Journal of Mathematical Economics}, 1(1):23--37, 1974.

\bibitem{solidity}
{Solidity documentation}.
\newblock \url{http://solidity.readthedocs.io/en/latest/index.html}.

\bibitem{WeikumV2001}
G.~Weikum and G.~Vossen.
\newblock {\em Transactional Information Systems: Theory, Algorithms, and the
  Practice of Concurrency Control and Recovery}.
\newblock Morgan Kaufmann Publishers Inc., San Francisco, CA, USA, 2001.

\bibitem{twophasecommit}
Wikipedia.
\newblock Two-phase commit protocol.
\newblock \url{https://en.wikipedia.org/wiki/Two-phase_commit_protocol}.
\newblock As of 18 May 2018.

\bibitem{Catalyst}
G.~Zyskind, C.~Kisagun, and C.~FromKnecht.
\newblock Enigma catalyst: a machine-based investing platform and
  infrastructure for crypto-assets.
\newblock \url{https://www.enigma.co/enigma_catalyst.pdf}.
\newblock As of 25 January 2018.

\end{thebibliography}
